\documentclass[draft,twoside]{article}
\newcommand{\mytitle}{The Hirota Identity for Hyperpfaffian $\tau$-Functions in Charge-$L$ Ensembles} 
\newcommand{\keywords}{Hyperpfaffian, Beta Ensembles, Hirota Identity}
\newcommand{\msc}{
15B52, 
60G55,  
82B23, 
15A15, 
}

\usepackage{amsmath,amssymb,
  amsthm,amsfonts,amscd,mathrsfs,pifont,upgreek} 
\usepackage{eucal}  
\usepackage{verbatim}
\usepackage{ifpdf}
\usepackage[colorlinks=true, citecolor=blue]{hyperref}

\ifpdf
        \usepackage[pdftex]{graphicx}
\else
        \usepackage[dvips]{graphicx}
\fi

\newtheorem{thm}{Theorem}[section]

\newtheorem{lemma}[thm]{Lemma}

\newtheorem{thm*}{Theorem}[]
\newtheorem{cor*}[thm*]{Corollary}
\newtheorem{claim*}[thm*]{Claim}
\newtheorem{lemma*}[thm*]{Lemma}
\newtheorem{prop*}[thm*]{Proposition}
\newtheorem{conj*}[thm*]{Conjecture}

\theoremstyle{definition}

\newtheorem{problem*}{Problem}[section]

\newtheorem{question*}{Question}[section]

\newtheorem{defn*}{Definition}

\theoremstyle{remark}
\newtheorem*{rem}{Remark}

\newcommand{\qq}[1]{\qquad \mbox{#1} \qquad}

\newcommand{\BB}[1]{\ensuremath{\mathbb{#1}}}

\newcommand{\R}{\ensuremath{\BB{R}}}

\DeclareMathOperator{\sgn}{sgn}

\DeclareMathOperator{\PF}{PF}

\usepackage[margin=1 in]{geometry}
\usepackage{lineno}

\numberwithin{equation}{section} 

\numberwithin{equation}{section}

\begin{document}
\title{\bfseries\sffamily \mytitle}  
\author{\sc Christopher D.~Sinclair}
\maketitle

\begin{abstract}
We study log-gas ensembles with inverse temperature $\beta = L^2$ using a
confluent Vandermonde representation that admits a formulation in the exterior
algebra of a finite-dimensional vector space. By interpreting the system as
consisting of finitely many particles with integer charge $L$, partition
functions can be expressed exactly as hyperpfaffians. In this formulation, the
system is governed by a natural momentum grading arising from the confluent
Vandermonde structure, and its statistical observables are determined entirely
by the corresponding bigraded commutative subalgebra. The geometric identity
that a particle's $L$-blade wedges with itself to zero produces momentum
Pl\"ucker relations within this algebra. These relations generate momentum
transport identities between sectors of different particle number. Upon
introducing dynamic time variables, the partition functions become
$\tau$-functions, and these transport identities are transformed into Hirota
bilinear equations. This provides an explicit algebraic origin for the
integrable hierarchy structure of the $\beta = L^2$ ensembles, which may be
viewed as a finite-dimensional analogue of the Sato Grassmannian formulation of
integrable systems.
\end{abstract}

{\bf MSC2010:} \msc

\vspace{1cm}

\section{Introduction}
\label{sec:intro}

This paper shows that the $\beta = L^2$ log-gas ensembles possess an integrable
structure: their partition functions are $\tau$-functions satisfying Hirota
bilinear equations. The key structural observation is that the confluent
Vandermonde determinant introduces another grading (which we call momentum) in
the exterior algebra formulation of the ensemble. This grading allows
computations to be carried out in a relatively low-dimensional commutative
subalgebra of the exterior algebra—the momentum algebra—whose structure
coefficients satisfy Plücker relations.

The integrable structure then follows from a purely algebraic mechanism. Each
particle is represented by an $L$-blade in the momentum algebra, and the
identity that a blade wedges with itself to zero produces Plücker relations.
These relations generate momentum transport identities between particle systems
with different particle numbers, and when written in dynamic Miwa coordinates
they take the form of Hirota bilinear equations. 

\subsection{Log-Gases and $\beta$-Ensembles}
\label{subsec:intro.log-gases}

Log-gases, or Coulomb gases, describe systems of particles on the real line
interacting through a logarithmic pair potential and confined by an external
field. The joint density of particle positions in such systems takes the form
\[
\frac{1}{Z_M}
\prod_{1 \le i < j \le M} |x_i - x_j|^{\beta}
\prod_{i=1}^M w(x_i),
\]
where $w(x)$ is a weight determined by the confining potential and $\beta > 0$
is the inverse temperature. These models play a central role in random matrix
theory, statistical mechanics, and integrable systems. For the classical values
$\beta = 1, 2, 4$, the ensembles exhibit Pfaffian or determinantal structure,
leading to exact formulas for partition functions and correlation functions in
terms of orthogonal or skew-orthogonal polynomials \cite{DeBruijn1955,
Dyson1962, Forrester2010, Mehta2004, TracyWidom1998}. For general $\beta$,
however, such structures are typically absent, and exact solvability becomes
more difficult \cite{DumitriuEdelman2002,RamirezRider2011}.

In this work we focus on the special class of ensembles with $\beta = L^2$.
These ensembles admit an interpretation in which each particle behaves as a
composite object built from $L$ fermionic constituents, leading naturally to
confluent Vandermonde determinants and an exterior algebra formulation. This
representation provides a natural geometric pathway to the exterior and momentum
algebra formulations developed here, circumventing the analytic obstacles
present for general $\beta.$

\subsection{Composite Particles and Confluent Vandermonde Determinants}
\label{subsec:intro.confluent-vandemonde}

When $\beta = L^2$ with $L \in \mathbb{N}$, the interaction term
\[
\prod_{1 \le i < j \le M} |x_j - x_i|^{L^2}
\]
can be interpreted as arising from a system in which each particle behaves like
a bound state of $L$ fermionic constituents, i.e., a charge-$L$ particle. In
this picture, each particle position $x_i$ is replaced by a cluster of $L$
coincident fermions, and the resulting antisymmetric structure produces higher
powers of the Vandermonde determinant.

Mathematically, this leads to confluent Vandermonde determinants, obtained by
replacing each particle coordinate by a jet of derivatives evaluated at that
coordinate \cite{Hou2015}. As a result, the partition function and correlation
functions of the $\beta = L^2$ ensemble can be written in terms of wedge
products built from these confluent Vandermonde structures. This representation
naturally places the model in an exterior algebra framework, which introduces
additional algebraic structure that will play a central role in the integrable
formulation developed in this paper.

\subsection{Exterior Algebra and Hyperpfaffian Partition Functions}
\label{subsec:intro.hyperpfaffians}

The confluent Vandermonde representation of the $\beta = L^2$ ensemble admits a
natural formulation in the exterior algebra of a finite-dimensional vector
space. In this formulation, each particle coordinate contributes an $L$-blade
constructed from derivatives of monomials evaluated at the particle position,
and an $M$-particle configuration is represented by the wedge product of these
blades. The partition function is obtained by evaluating the top-degree
component of a wedge power of a single background element encoding the moments
of the weight function.

This construction leads naturally to hyperpfaffian structure: the partition
function can be expressed as a hyperpfaffian of a tensor built from the moments
of the measure weight, generalizing the determinantal and Pfaffian formulas that
appear in the classical cases $\beta = 2$ and $\beta = 1,4$
\cite{Barvinok1995,LuqueThibon2002,Sinclair2012,Ishikawa1995}. This reduces the
many-body integral defining the partition function to a purely algebraic problem
in the exterior algebra and makes the combinatorial structure of the ensemble
explicit.

\subsection{The Momentum Algebra and Dimensional Reduction}
\label{subsec:intro.momentum-algebra}

The confluent Vandermonde construction introduces an additional grading that
plays a central role in the algebraic structure of the $\beta = L^2$ ensembles.
When a particle is replaced by a cluster of $L$ fermionic constituents, the
resulting $L$-blade is built from the Wronskian of $L$ monomials. The Wronskian
is formed from derivatives of monomials evaluated at the positions of the
fermionic constituents. These derivatives introduce degree shifts, and after
recentering the indices, each basis element in the exterior algebra may be
assigned an integer label which we interpret as a {\em momentum}. Binding $L$
fermions into a composite particle therefore introduces a new conserved quantum
number that records the degree shift relative to a centered reference
configuration.

Momentum is additive under wedge products, and the structure of the confluent
Vandermonde determinant imposes strict bounds on the allowed total momentum in
each particle sector. Consequently, only a finite range of momentum modes can
appear in the wedge products relevant to the partition function and correlation
functions. This leads to a bigraded commutative subalgebra of the exterior
algebra, which we call the {\em momentum algebra}. The observables of the
$\beta = L^2$ ensemble, including partition functions, correlation functions,
and transport relations, can be expressed entirely in terms of structure
coefficients in the momentum algebra. This construction reduces the complexity
of encoding particle backgrounds from ${LM \choose L}$ coefficients in the
exterior algebra to only $L^2(M-1)+1$ momentum coefficients in the momentum
algebra \cite{SinclairWells}.

\subsection{Momentum Pl\"ucker Relations and Transport Identities}
\label{subsec:intro.momentum-plucker}

The exterior algebra formulation implies that the basic $L$-blade associated
with a charge-$L$ particle satisfies the identity
\[
\omega(z) \wedge \omega(z) = 0,
\]
since the wedge product of any blade with itself vanishes. When expressed in the
momentum algebra, this identity produces a family of quadratic relations among
the generators of the momentum algebra, which we refer to as the {\em momentum
Pl\"ucker relations}. These relations are analogues of the classical Pl\"ucker
relations describing the Grassmannian \cite{Sato1981,MiwaJimboDate2000}, but in
the present setting they are organized by momentum conservation: each relation
corresponds to a fixed total momentum and expresses a linear dependence among
wedge products whose momenta sum to the same value.

These momentum Pl\"ucker relations impose algebraic constraints on the structure
coefficients that govern partition functions and correlation functions in the
momentum algebra. When interpreted in the particle sectors of the exterior
algebra, these constraints take the form of {\em momentum transport identities},
which describe the transport of momentum between backgrounds with different
particle numbers. These transport identities will be expressed in generating
function form where they become Hirota bilinear equations for the associated
$\tau$-functions.

\subsection{$\tau$-Functions, Miwa Coordinates and the Hirota Identity}
\label{subsec:intro.tau-miwa-hirota}

The momentum transport identities derived from the Pl\"ucker relations may be
expressed naturally in generating function form by deforming the weight function
with auxiliary time variables \cite{Miwa1982,MiwaJimboDate2000}. In this
formulation, the partition functions of charge-$L$ ensembles become $\tau$-functions
depending on an infinite sequence of times, and shifts in these times correspond
to the insertion or extraction of particles with specified momentum. Introducing
Miwa coordinates allows these time deformations to be encoded in a spectral
parameter, while an additional fugacity parameter tracks the number of particle
extractions necessary for momentum transport. In these coordinates, the momentum
transport identities take the form of bilinear residue relations for the
associated Baker--Akhiezer wave functions, yielding a Hirota bilinear equation
for the $\tau$-functions of the $\beta = L^2$ ensemble.

In the classical theory, Hirota equations arise from Plücker relations on the
infinite Grassmannian \cite{AdlerVanMoerbeke2001,HarnadOrlov2021,MiwaJimboDate2000,Sato1981}.
In the present setting, the Hirota equations instead arise from momentum Plücker
relations in a finite-dimensional exterior algebra determined by the confluent
Vandermonde structure. The integrable structure of the $\beta=L^2$ ensembles may
therefore be viewed as a finite-dimensional analogue of Sato theory.

\subsection{The Main Results}
\label{subsec:intro.main-results}

The main result of this work is the identification of a `small' commutative
subalgebra of the exterior algebra with an additional momentum grading
underlying $\beta = L^2$ ensembles, and a family of bilinear momentum transport
identities arising from Pl\"ucker relations in this algebra. When expressed in
Miwa coordinates, these transport identities become Hirota bilinear equations
for the associated $\tau$-functions formed from hyperpfaffian partition
functions, revealing an integrable structure for the $\beta = L^2$ ensembles
analogous to the KP and BKP hierarchies that appear in the classical cases
$\beta = 2$ and $\beta = 1,4$.

We emphasize that in this work we do not solve the $\beta=L^2$ ensembles in the
sense of obtaining explicit formulas for correlation kernels or asymptotics.
Rather, we establish that these ensembles are {\em integrable} in the sense that their
partition functions are $\tau$-functions satisfying Hirota bilinear equations.
This shows that the ensembles are solvable at the level of integrable systems structure,
though the explicit evaluation of the resulting $\tau$-functions remains a
separate problem.

The structure of the paper is as follows. In Section~\ref{sec:ensembles} we
define the $\beta = L^2$ log-gas model. In
section~\ref{sec:confluent-vandermonde} we review the confluent Vandermonde
representation. In Section~\ref{sec:exterior-algebra} we develop the exterior
algebra formulation and derive the hyperpfaffian expression for the partition
function. Everything through Section~\ref{sec:exterior-algebra} already appears
in the literature \cite{SinclairWells}. 

Section~\ref{sec:momentum-algebra} introduces the momentum algebra and the
associated dimensional reduction. In Section~\ref{sec:insertion-extraction} we
derive the momentum Pl\"ucker relations and the resulting transport identities.
Section~\ref{sec:tau-functions} introduces $\tau$-functions and Miwa
coordinates, and in Section~\ref{sec:hirota} we derive the Hirota bilinear
identity. We conclude in Section~\ref{sec:discussion} with a discussion of
connections to integrable hierarchies, odd $L$, correlation kernels, and
asymptotic questions.

\section{Ensembles of Charge-$L$ Particles}
\label{sec:ensembles}

We consider an ensemble of $M$ charge-$L$ particles interacting logarithmically
on the real line $\R$ in the presence of an external potential. We refer to this as the {\em
charge-$L$ log-gas ensemble.} While these ensembles are equivalent to the
$\beta$-ensembles with $\beta=L^2$, but the charge interpretation will be
convenient for the algebraic formulation developed later. We fix notation for
the Gibbs measure, partition function, and correlation functions, which will be
rewritten in subsequent sections using confluent Vandermonde determinants and
exterior algebra.

Given two particles located at $x, x' \in \R$ their
{\em pairwise interaction energy} is given by
\[
E(x, x') = -L^2 \log|x - x'|.
\]
In our interpretation, the factor $L^2$ arises as the coupling constant
representing interaction strength between two charge-$L$ particles. Given $M$
particles at locations $\mathbf x = (x_1, \ldots, x_M)$ the interaction energy
is the sum of pairwise interaction energies,
\[
E(\mathbf x) = -\sum_{1 \leq i < j 
\leq M} L^2 \log|x_j - x_i|.
\]
To confine the particles we introduce a potential $V : \R \rightarrow [0,
\infty),$ so that the total energy of a system of $M$ particles is given by
\[
E_{\mathrm{tot}}(\mathbf x) = E(\mathbf x) + \sum_{i=1}^M V(x_i).
\]

We normalize the inverse temperature $(kT)^{-1} = 1,$ and the {\em Boltzmann
factor} of the ensemble is given by
\[
e^{- (kT)^{-1} E_{\mathrm{tot}}(\mathbf x)} = 
\bigg\{ \prod_{1 \leq i < j \leq M} | x_j - x_i |^{L^2} \bigg\} \bigg\{\prod_{i=1}^M e^{-V(x_i)}\bigg\}
\]
Since the Gibbs factor depends only on the product $(kT)^{-1} E_{\mathrm{tot}}$,
the interaction exponent $L^2$ can be interpreted in two equivalent ways: either
as an inverse temperature $\beta = L^2$ for particles of unit charge, or as unit
inverse temperature for particles carrying charge $L$. This provides a natural
explanation for the special values $\beta = L^2$: they correspond to Coulomb
systems in which the particles carry integer charge. In this interpretation, the
classical orthogonal ($\beta=1$) and symplectic ($\beta=4$) ensembles correspond
to charge-$1$ and charge-$2$ ensembles within this family, and their Pfaffian
solvability appears as the $L = 1$ and $L = 2$ cases of the hyperpfaffian
solvability discussed here. The classical unitary ($\beta = 2$) ensembles instead belong
to a distinct determinantal solvability class.

If we define the weight $w(x) = e^{-V(x)}$ and the corresponding measure
$\mu(dx) = e^{-V(x)} dx,$ then the {\em Gibbs measure} of the ensemble is given
by
\[
\frac{1}{M! Z} \prod_{1 \leq i < j \leq M} | x_j - x_i |^{L^2} \, d\mu^M(\mathbf x)
\]
where $\mu^M$ is the $M$-fold product measure of $\mu$ on $\R^M,$ and
\[
Z = \frac{1}{M!} \int_{\R^M} \prod_{1 \leq i < j \leq M} | x_j - x_i |^{L^2} \, d\mu^M(\mathbf x)
\]
is the {\em partition function} for the ensemble \cite{Forrester2010,
Mehta2004}. We will see that, once dynamic coordinates (times) are introduced,
the partition function naturally acquires the structure of a $\tau$-function.

\subsection{Correlation Functions}
\label{subsec:ensembles.correlation-functions}

While the partition function captures the global statistics of the ensemble,
local particle localization is described by the correlation functions. For $1
\leq m \leq M,$ the $m$th {\em correlation function} $R_m : \R^m \rightarrow [0,
\infty)$ gives the joint marginal density for finding $m$ particles at specified
positions $x_1, \ldots, x_m$, defined as
\begin{align*}
R_m(x_1, \ldots, x_m) &= \frac{1}{(M-m)! Z} \bigg\{ \prod_{i=1}^m w(x_i) \bigg\}\bigg\{\prod_{1 \leq i < j \leq m} |x_j - x_i|^{L^2} \bigg\} \\ & \qquad \times \int_{\R^{M-m}}  \bigg\{ \prod_{1 \leq a < b \leq M-m}  |y_b - y_a|^{L^2} \bigg\} \bigg\{ \prod_{a=1}^{M-m} \prod_{i=1}^m |x_i - y_a|^{L^2} \bigg\} d\mu^{M-m}(\mathbf y).  
\end{align*}
A central goal of this work is to uncover the algebraic mechanism underlying
these local statistics. In later sections, we will reinterpret these dense
integrals using exterior algebra, where evaluating a correlation function
corresponds precisely to applying algebraic insertion operators to a fixed
background state.

In particular, the first correlation function 
\[
R_1(x) = \frac{w(x)}{(M-1)! Z} \int_{\R^{M-1}}  \bigg\{ \prod_{1 \leq a < b \leq M-1} | y_b - y_a |^{L^2} \bigg\}
\bigg\{ \prod_{a=1}^{M-1} |x - y_a|^{L^2} \bigg\} d\mu^{M-1}(\mathbf y),
\]
will later be expressed algebraically in terms of insertion operators in the
exterior algebra formulation.

The key observation underlying this work is that when $\beta = L^2$, the interaction factors
$|x_k - x_i|^{L^2}$ admit an interpretation in which each particle is replaced
by a cluster of $L$ fermionic constituents. This ``clustering'' allows the
partition function and correlation functions to be expressed in terms of confluent
Vandermonde determinants, providing the essential bridge between the statistical
mechanics of the ensemble and the exterior algebra.

\section{The Confluent Vandermonde Determinant}
\label{sec:confluent-vandermonde}

In the classical $L=1$ case, the interaction term in the Gibbs measure is given
by the absolute value of the standard Vandermonde determinant. For general charge-$L$ particles,
this interaction becomes a higher power, which we realize algebraically using a
confluent Vandermonde construction. Rather than simple scalar entries, this representation
encodes particle configurations using blocks built from derivatives of Vandermonde
columns. These matrix blocks serve as the concrete coordinate representation for
the abstract $L$-forms developed in the subsequent exterior algebra framework.

The $M \times M$ {\em Vandermonde determinant} identity is 
\[
\det \begin{bmatrix}
1 & 1 &  & 1 \\
x_1 & x_2 &  & x_M \\
x_1^2 & x_2^2 & \cdots & x_M^2 \\
& \vdots & \ddots & \vdots \\
x_1^{M-1} & x_2^{M-1} & \cdots & x_M^{M-1}
\end{bmatrix} = \prod_{1 \leq i < j \leq M} (x_j - x_i).
\]
Up to absolute values, this is the interaction term in the Gibbs measure for the
$L=1$ (classical orthogonal) ensembles. Note the identification between particle
locations and column vectors of the form $[1, x, x^2, \ldots, x^{M-1}]^T$. In
this representation, a charge-1 particle located at $x$ is encoded by a single
Vandermonde column.

In order to encode the locations of charge $L$ particles we need a version of
the {\em confluent} Vandermonde determinant. In this situation we start with a
standard Vandermonde column $v(x) = [1, x, x^2, \ldots, x^{LM - 1}]^T \in
\R^{LM}$ and define the renormalized differentiation operators $D^{\ell}  =
\frac{1}{\ell!} \frac{d^{\ell}}{dx^{\ell}}.$ We define the differentiated
Vandermonde columns $D^{\ell} v(x)$ where $D^{\ell}$ is applied to each entry in
$v(x).$ The $LM \times L$ matrix, 
\[
\mathbf V_L(x) = [ v(x) \; D^1 v(x) \; \cdots \; D^{L-1} v(x)]
\]
is a confluent Vandermonde {\em block}. In the confluent Vandermonde
representation, a charge-$L$ particle located at $x$ is represented by the
$L$-dimensional column block $\mathbf V_L(x)$. We use these to construct the
confluent $LM \times LM$ Vandermonde matrix
\[
\mathbf V( \mathbf x ) = \left[ \mathbf V_L(x_1) \; \mathbf V_L(x_2) \; \cdots \; \mathbf V_L(x_M) \right].
\]
The confluent Vandermonde determinant identity is then
\[
\det \mathbf V(\mathbf x) = \prod_{i < j}^M (x_j - x_i)^{L^2}.
\]
Up to absolute values, this is the interaction term in the Gibbs measure of the
charge $L$ ensemble, and in this situation we may identify a charge-$L$ particle
at $x$ with the confluent Vandermonde block $\mathbf V_L(x).$ 

The {\em Pl\"ucker coordinates} for an $ML \times L$ block matrix consist of all
determinants of $L \times L$ minors of the block matrix. If we index $L$ rows of
$\mathbf V_L(x)$ by $0 \leq r_1 < r_2 < \cdots < r_L < LM$ then the associated
Pl\"ucker coordinate is
\[
\det \begin{bmatrix}
x^{r_1} & D^1 x^{r_1} &        & D^{L-1} x^{r_1} \\
x^{r_2} & D^1 x^{r_2} & \cdots & D^{L-1} x^{r_2} \\
        & \vdots      & \ddots & \vdots \\
x^{r_L} & D^1 x^{r_L} & \cdots & D^{L-1} x^{r_L} \\
\end{bmatrix} = \prod_{1 \leq i < j \leq L} (r_j - r_i) \cdot x^{r_1 + \cdots + r_L - {L \choose 2}}.
\]
This is the renormalized {\em Wronskian} of the monomials $x^{r_1}, x^{r_2},
\ldots, x^{r_L}.$ The exponent $r_1 + \cdots + r_L - {L \choose 2}$ will play an
important role later; after recentering indices, it provides the momentum
grading that underlies the momentum algebra introduced in Section~\ref{sec:momentum-algebra}.

The confluent Vandermonde determinant therefore provides an exact algebraic
encoding of the interaction term, organizing the fermionic constituents into $L
\times L$ matrix blocks. This suggests a direct translation of the entire
statistical ensemble into the geometry of an exterior algebra.

\section{The Exterior Algebra}
\label{sec:exterior-algebra}

By interpreting the confluent Vandermonde blocks as $L$-forms, we can 
completely reformulate the charge-$L$ ensemble within an exterior algebra. 
In this space, an $M$-particle configuration is encoded simply as the wedge 
product of $M$ such forms. The physical act of integrating over particle 
coordinates is replaced by an algebraic evaluation in the top degree (the 
determinantal line). Crucially, this allows the many-body integral to factorize: 
the partition function becomes the top-degree component of a wedge power of 
a single background element. This factorization is the exact algebraic mechanism 
that produces the hyperpfaffian structure.

Let $I$ be $LM$ indices centered at 0,
\[
I = \left\{ -\frac{LM-1}2, -\frac{LM-3}2, \ldots, \frac{LM-3}2, \frac{LM-1}2 \right\}.
\]
Note the indices are integers when $LM$ is odd and half-integers when $LM$ is
even. The centered fermionic index $r \in I$ corresponds to the original
monomial degree 
\[
p = r + c_M, \qq{where} c_M = \frac{(LM-1)}2
\] 
shifts indices so they are symmetric about 0. Centering the indices is not
essential but simplifies later formulas by making momentum indices symmetric
about zero.

The exterior algebra will serve as the algebraic space in which particle
configurations and their interactions are encoded. Let $V$ be a real vector
space of dimension $LM$ with basis $\{\mathbf e_i : i \in I\}.$ The exterior
algebra over $V$ is denoted 
\[
\Lambda V = \bigoplus_{m=0}^{LM} \Lambda^m V,
\] 
where $\Lambda^1 V = V$ and $\Lambda^m V$ is the linear space of $m$-{\em forms} with induced basis
\[
\bigg\{ \mathbf e_J = \mathbf e_{j_1} \wedge \mathbf e_{j_2} \wedge \cdots \wedge \mathbf e_{j_m} : J = \{j_1 < j_2 < \ldots < j_m\} \in {I \choose m} \bigg\},
\]
and ${I \choose m}$ is the collection of subsets of $I$ of cardinality $m$.
Elements of the form $u_1 \wedge u_2 \wedge \cdots \wedge u_m$ for vectors $u_1, \ldots, u_m \in V$ are called $m$-{\em blades}. 

The wedge product is alternating. That is, given an $m$-form $\alpha$ and an
$n$-form $\eta$, we have
\[
\alpha \wedge \eta = (-1)^{nm} \eta \wedge \alpha.
\]

The top grading $\Lambda^{LM} V$ is one-dimensional, called the {\em
determinantal line}, with distinguished basis element $\mathbf e_I$ which we
call the {\em volume form}. Given $\alpha \in \Lambda V$ we write $\star \alpha$
for the coefficient of the volume form, that is $\star : \Lambda V \rightarrow
\R$ is the distinguished linear functional supported on the determinantal line
with $\star \mathbf e_I = 1.$ Wedge powers of $L$-forms will encode
configurations of $M$ charge-$L$ particles, and {\em evaluation} via $\star$ on the
determinantal line will produce partition functions and correlation functions.

When $L$ is even (a situation we will soon restrict ourselves to) and $\alpha \in \Lambda^L V$ then 
\[
\alpha^{\wedge M} = \underbrace{\alpha \wedge \alpha \wedge \cdots \wedge \alpha}_M \in \Lambda^{LM} V,
\]
and we define the {\em hyperpfaffian} of $\alpha$ to be
\[
\PF(\alpha) := \star \frac{\alpha^{\wedge M}}{M!} \in \R.
\]
This algebraic definition generalizes the Pfaffian to higher-degree forms
\cite{Barvinok1995,Ishikawa1995}. When $L=2$, the hyperpfaffian generalizes
the notion of the Pfaffian of an antisymmetric matrix. If $\mathbf A =
[a_{j,k}]_{j,k \in I}$ is an antisymmetric matrix, then we may associate the
2-form 
\[
\alpha = \sum_{j < k} a_{j,k} \mathbf e_j \wedge \mathbf e_k, 
\]
and the classical {\em Pfaffian} of $\mathbf A$ is the hyperpfaffian of $\alpha.$ 

\subsection{Representing Particles in the Exterior Algebra}
\label{subsec:exterior-algebra.particles}

$\Lambda V$ is the algebraic space which encodes configurations of particles in
an ensemble with total charge $LM$. Our ultimate goal is to study ensembles
consisting of $M$ charge-$L$ particles, but the fermionic starting point is a
system of $LM$ charge-$1$ particles. A fermion located at $y$ is identified with the 1-form
\[
v(y) = \sum_{r \in I} y^{r + c_M} \mathbf e_r,
\]
that is, the algebraic avatar of a charge 1 particle at $y$ is the Vandermonde
column $v(y)$. A collection of $L$ fermions located at $y_1, y_2, \ldots, y_L$
is then encoded in the $L$-form
\[
v(y_1) \wedge v(y_2) \wedge \cdots \wedge v(y_L),
\]
and we seek to bind these $L$ fermionic constituents into a single charge $L$
particle located at $x$ by taking an appropriate limit as $y_1, y_2, \ldots, y_L
\rightarrow x.$ Doing this naively produces 0, but if we divide by the
interaction term between the $L$ fermions and then take the limit we get exactly
the $L$-form associated to the confluent Vandermonde block $\mathbf V_L(x).$
That is,
\[
\lim_{y_1, \ldots, y_L \rightarrow x} \frac{v(y_1) \wedge v(y_2) \wedge \cdots \wedge v(y_L)}{\prod_{1 \leq j<k \leq L}(y_k - y_j)} = v(x) \wedge D^1 v(x) \wedge \cdots \wedge D^{L-1} v(x)
\]
This procedure corresponds to forcing $L$ fermions to occupy the same location,
and the normalization by the interaction term removes the vanishing caused by
the antisymmetry of the wedge product. This process is called {\em
confluentization} and it provides an algebraic avatar for a charge $L$-particle at $x$. 
We therefore define
\begin{equation}
\label{eq:omega}
\omega(x) := v(x) \wedge D^1 v(x) \wedge \cdots \wedge D^{L-1} v(x)
\end{equation}
Because $\omega(x)$ is the wedge of $L$ 1-forms, it is an $L$-blade. The
exterior algebra naturally provides a fermionic description of the system, and
charge-$L$ particles arise by binding together $L$ fermionic constituents
through the confluent limit into an $L$-blade.

In the exterior algebra, the Vandermonde determinant identity is
\[
\star \big( v(x_1) \wedge v(x_2) \wedge \cdots \wedge v(x_{LM}) \big) = \prod_{1 \leq i < k \leq LM} (x_k - x_i),
\]
and the confluent Vandermonde determinant is
\[
\star \big( \omega(x_1) \wedge \omega(x_2) \wedge \cdots \wedge \omega(x_M) \big) = \prod_{1 \leq i < k \leq M} (x_k - x_i)^{L^2}. 
\]

The interaction terms in the Gibbs measure therefore translate directly into
evaluations of wedge products in the exterior algebra. Because the confluent
Vandermonde blocks correspond exactly to $L$-forms, the physical observables of
the ensemble---such as partition functions and correlation functions---can now
be recast as purely algebraic operations.

\subsection{Hyperpfaffian Partition Functions}
\label{subsec:exterior-algebra.partition-function}

With the interaction term encoded algebraically, the partition function takes
the form of an integrated wedge product of $M$ distinct $L$-forms. The crucial
algebraic mechanism here is factorization. Because the integrands wedge-commute,
the many-body integral passes through the wedge product to act on a single
$L$-form. This reduces the partition function to the top-degree component of a
wedge power of this single integrated form, leading directly to a hyperpfaffian
expression.

The partition function in the exterior algebra setting is given by
\[
Z = \frac{1}{M!} \int_{\R^M} \left| \star ( \omega(x_1) \wedge \omega(x_2) \wedge \cdots \wedge \omega(x_M) ) \right| d\mu^M(\mathbf x).
\]
When $L$ is even the absolute values are redundant, and when $L$ is odd we can introduce an alternating term which compensates for the signs inside the absolute values:
\begin{align*}
Z &= \frac{1}{M!} \int_{\R^M} \star ( \omega(x_1) \wedge \omega(x_2) \wedge \cdots \wedge \omega(x_M) ) d\mu^M(\mathbf x) \qquad & \mbox{$L$ even;} \\
Z &= \frac{1}{M!} \int_{\R^M} \star ( \omega(x_1) \wedge \omega(x_2) \wedge \cdots \wedge \omega(x_M) ) \prod_{1 \leq j < k \leq M} \sgn(x_k - x_j) \cdot d\mu^M(\mathbf x) \qquad & \mbox{$L$ odd.} 
\end{align*}
When $L$ is even, the wedge product of the $L$-forms is symmetric, and the
absolute values in the partition function may be removed. When $L$ is odd,
additional sign factors appear and must be compensated by inserting an
alternating term. Since this introduces additional algebraic complications that
are not central to the present work, we restrict attention to even $L$ from this
point forward.

When $L$ is even,
\[
Z = \star \frac{1}{M!}  \gamma^{\wedge M} = \PF\left(\gamma \right) \qq{where} \gamma = \int_{\R} \omega(x) \, d\mu(x).
\]
Here, by integrating the $L$-form $\omega(x)$ we mean: produce a new $L$-form,
the {\em Gram form} $\gamma$ whose Pl\"ucker coordinates are the integrals of
the Pl\"ucker coordinates of $\omega(x).$ In this sense, the partition function
of the charge-$L$ ensemble is the hyperpfaffian of a Gram form, just as the
partition function of classical ensembles is the determinant or Pfaffian of a
Gram matrix formed from the moments (or skew-moments) of the weight of the ensemble.

The identity expressing the partition function as the top-degree component of a
wedge power is sometimes referred to as Chen’s lemma on iterated integrals
\cite{LuqueThibon2002}. In the present setting, where the integrands
wedge-commute, this identity reduces to an application of Fubini’s theorem
together with antisymmetry of the wedge product, allowing the many-body integral
to factor as a wedge power of a single integrated $L$-form \cite{Sinclair2012}.

In our centered indices, the Pl\"ucker coordinates of $\omega(x)$ are given by
\[
\mathrm{Wr}_J(x) = \prod_{1 \leq i < k \leq L} (j_k - j_i) x^{j_1 + \cdots + j_L + c_M }, \qquad J = \{j_1, j_2, \ldots, j_L\} \in {I \choose L}.
\]
Then, the Pl\"ucker coordinate of $\gamma$ with index $J$ is given by
\begin{align*}
\mathrm{Gr}_{J} &= \int_{\R} \mathrm{Wr}_J(x) \, d\mu(x) \\
&= \prod_{1 \leq i < k \leq L} (j_k - j_i) \cdot \int_{\R} x^{j_1 + \cdots + j_L + c_M} \, d\mu(x).
\end{align*}
If we designate the $k$th moment of $\mu$ to be $m_k,$ and the $k$th shifted moment by $\widehat m_k = m_{k+c_M}$ then we can explicitly write
\[
\gamma = \sum_{J \in {I \choose L}} \bigg\{ \prod_{1 \leq i < k \leq L} (j_k - j_i) \cdot \widehat m_{j_1 + \cdots + j_L} \bigg\} \mathbf e_J.
\]

We conclude that the partition function can be computed algebraically by forming
the Gram form and computing its hyperpfaffian. Remarkably, this
single algebraic object encodes not just the global partition function, but all
local particle statistics as well.

\subsection{Correlation Functions}
\label{subsec:exterior-algebra.correlation-functions}

In the exterior algebra framework, evaluating a correlation function corresponds
simply to inserting fixed particles into a fixed background determined by the
weight of the ensemble. By applying the same factorization identity used for the
partition function, we can express these local densities by wedging the
$L$-forms associated with the fixed particle locations directly into the
remaining wedge power of the Gram form.

\begin{lemma}\label{lemma:correlations}
For $1 \leq m \leq M,$ 
\[
R_m(x_1, \ldots, x_m) = \frac{1}{Z} \bigg\{ \prod_{i=1}^m w(x_i) \bigg\} \star\left(\omega(x_1) \wedge \cdots \wedge \omega(x_m) \wedge \frac{\gamma^{\wedge(M-m)}}{(M-m)!}\right).
\] 
In particular, 
\begin{equation}
\label{eq:R1}
R_1(x) = \frac{w(x)}{Z} \star \left( \omega(x) \wedge \frac{\gamma^{\wedge(M-1)}}{(M-1)!} \right).
\end{equation}
\end{lemma}
\begin{proof}
This follows from the same wedge-power identity used for the partition function: integrating over the remaining $M-m$ particle positions produces the $(M-m)$th wedge power of the Gram form, while the fixed particle positions contribute the wedge factors $\omega(x_i).$ See \cite{SinclairWells} for details.
\end{proof}

This formula demonstrates that local correlation functions are obtained by
inserting the forms corresponding to fixed particles into the 
background determined by the Gram form. Consequently, the Gram form encodes the
entire statistical structure of the ensemble. While this insertion-background
decomposition is conceptually elegant, direct computation remains formidable:
these objects live in the exterior algebra $\Lambda^L V$, a space whose
dimension $\binom{LM}{L}$ grows combinatorially with $M$. This computational
bottleneck motivates the search for a structural symmetry that can reduce the
effective dimension of the problem.

\section{The Momentum Algebra}
\label{sec:momentum-algebra}

The combinatorial explosion of the full exterior algebra can be bypassed by
exploiting the internal structure of the confluent Vandermonde determinant.
Because each charge-$L$ particle is formed by binding $L$ fermionic
constituents, the resulting $L$-forms appear only in highly restricted linear
combinations indexed by their total momentum. This physical constraint carves
out a finite-dimensional, bigraded commutative subalgebra generated strictly by
these momentum modes. We refer to this space as the {\em momentum algebra}.
Restricting our focus to this algebra provides a massive dimensional reduction
in which partition functions, correlation functions, and transport identities
can be expressed and computed more efficiently.

The fact that even $L$-forms commute means that charge $L$-particles satisfy bosonic statistics. 
For integer $p$ define the $p$th {\em momentum mode}
\[
\epsilon_p = \sum_{J \in {I \choose L} \atop j_1 + \cdots + j_L = p} \prod_{1 \leq i < k \leq L} (j_k - j_i) \cdot \mathbf e_J \quad \in \Lambda^L V.
\]
The only non-zero momentum modes satisfy $|p| \leq c_M.$ We define $\mathcal S =
\mathrm{span}_{\R} \{ \epsilon_p \}$ to be the {\em momentum spine,} and in this
basis,  
\begin{equation}
\label{eq:omega-generating}
\omega(x) = \sum_{p} x^{c_M + p} \epsilon_p
\end{equation}
and
\begin{equation}
\label{eq:gamma}
\gamma = \sum_{p} \widehat m_p \epsilon_p.
\end{equation}
Note that the power $c_M$ ensures that only
non-negative powers of $x$ appear in $\omega(x).$ We will often abbreviate sums
over momentum modes using $\sum_p \cdots,$ with the implicit understanding if
$|p| > c_M$ we take $\epsilon_p = 0.$ 

Let us explain the terminology. We view $J = \{j_1 < j_2 < \ldots < j_L\} \in I$ as
quantum numbers representing {\em momentum} for the $L$ bound fermions. The sum $j_1
+ j_2 + \cdots +  j_L$ is the {\em momentum mode} of a charge $L$ particle
formed from the bound fermions. We may think of the $L \times L$ Vandermonde determinant 
\[
\prod_{1 \leq i < k \leq L} (j_k - j_i)
\]
as the {\em statistical weight} contributed by $L$ fermions in this particular
momentum state, and $\epsilon_p$ as the superposition of all fermionic states
producing a charge $L$ particles with total momentum $p.$ In this interpretation
we can think of $\gamma$ as a delocalized particle, averaged over $\mu.$ The
shifted moments $\widehat m_p$ then represent the {\em amplitude} of the
momentum modes $p$ of $\gamma.$ Amplitude here is not amplitude in the sense of
classical quantum mechanics, but it serves the same role with modulus (or norm)
replaced by a hyperpfaffian calculation. 

The momentum spine generates the graded commutative subalgebra of $\Lambda V,$ 
\[
\mathcal A := \bigoplus_{m=0}^M \mathcal A^m,
\]
where
\[
\mathcal A^0 = \R, \quad \mathcal A^1 = \mathcal S \qq{and} \mathcal A^m = \mathrm{span}_{\R}\{ \epsilon_{p_1} \wedge \cdots \wedge \epsilon_{p_m} :  p_1, \ldots, p_m \}.
\]
Note that $\mathcal A^M = \Lambda^{LM} V.$ In other words, the top grading of the momentum
algebra coincides with the determinantal line of the full exterior algebra.
Consequently, the linear functional $\star$ used to evaluate partition functions
and correlation functions may be computed entirely within the momentum algebra.

The $\epsilon_p$ commute so in general $\epsilon_p \wedge \epsilon_p \neq 0.$
Indeed, in general $\epsilon_p \wedge \epsilon_p$ will produce a non-trivial
element of $\mathcal A^2$ with momentum $2p.$ This is a distinction from the
fermionic situation: two fermions cannot be in the same momentum state, but two
charge $L$ particles can. In spite of this distinction, the momentum algebra
plays a similar role to that of $\Lambda V$ in the fermionic case. 

We call generic elements of $\mathcal A$ {\em backgrounds.} Elements of $\mathcal A^m$ are $m$-particle backgrounds, and elements of $\mathcal A^M$ are {\em saturated} backgrounds. An $m$-particle background of the form $\eta^{\wedge m}/m!,$ for $\eta \in \mathcal S,$ is called {\em primary}. The background $\gamma^{\wedge M}/M!$ is the (primary, saturated) {\em Gram background} of the ensemble.

Momentum gives a second grading on $\mathcal A,$
\[
\mathcal A = \bigoplus_{m=0}^M \bigoplus_p \mathcal A^m_p
\]
where $\mathcal A^m_p$ consists of $m$ particle backgrounds with (total) momentum $p,$ 
\[
\mathcal A^m_p = \mathrm{span}_{\R}\{ \epsilon_{j_1} \wedge \cdots \wedge \epsilon_{j_m} : j_1 + \cdots + j_m = p \}.
\]
We define 
\[
\pi_p : \mathcal A \rightarrow \mathcal A_p := \bigoplus_{m=0}^M \mathcal A^{m}_p, \qquad \pi^m : \mathcal A \rightarrow \mathcal A^m \qq{and} \pi_p^m : \mathcal A \rightarrow \mathcal A_p^m.
\]
to be respectively the projection of a background onto $\mathcal A_p,$ $\mathcal A^m$ and $\mathcal A_p^m.$ 

The non-empty momentum modes in $\mathcal A^m$ depend on $m$. In particular, the
highest grading has only one possible value for momentum: zero. That is,
$\mathcal A^M = \mathcal A^M_0.$ This follows because an $M$ particle background
is the superposition of states of $LM$ fermions bound into groups of $L.$
Whether fermions are bound or not, they must have distinct quantum numbers.
These quantum numbers are symmetric about the origin, and so the total momentum
of any saturated background is zero. 

This gives us a selection mechanism. If $\alpha$ is an $m$-particle background
with momentum $p$ and $\eta$ is an $(M-m)$-particle background with momentum $q$,
then $\alpha \wedge \eta$ is a saturated background only if $p + q = 0.$
Otherwise it will be zero. That is, $\star( \alpha \wedge \eta) \neq 0$ only if
$\alpha$ and $\eta$ have complementary momentum {\em and} particle numbers. This
provides a strict conservation law: only backgrounds whose particle number and momentum
complement each other to produce a saturated background contribute to
evaluations on the determinantal line.

\subsection{Structure Coefficients}
\label{subsec:momentum-algebra.structure-coefficients}

Because evaluations on the determinantal line vanish unless momentum is conserved, 
we can compute the partition function
\[
Z = \star \frac{1}{M!} \gamma^{\wedge M}
\]
entirely within the momentum algebra $\mathcal A.$ This provides a massive 
computational reduction: while the Gram form $\gamma \in \Lambda^L V$ requires 
${LM \choose L}$ Pl\"ucker coefficients, its representation in $\mathcal A$ is 
completely determined by only $2c_M + 1 = L^2(M-1) + 1$ momentum coefficients 
\cite{SinclairWells}.

To define the universal structure coefficients, let $K = \{p : |p| \leq c_M \}$ denote the allowed momentum modes of a single particle. Then $K^M$ represents all possible momentum configurations (ordered by particle index). Define
\[
K^M_0 := \{ P = (p_1, \ldots, p_M) \in K^M : p_1 + \cdots + p_M = 0 \}, \qq{and} \epsilon_P = \epsilon_{p_1} \wedge \epsilon_{p_2} \wedge \cdots \wedge \epsilon_{p_M}.
\]
\begin{lemma}
Define the structure coefficients,
\[
C_{P} := \star \epsilon_P, \quad P \in K^M_0.
\]
Then the partition function reduces to
\[
Z = \frac{1}{M!} \sum_{P \in K^M_0} C_{P} \prod_{i=1}^M \widehat m_{p_i}.
\]
That is, the partition function is a homogenerous polynomial in the moments of
$\mu$ with universal coefficients given by the $C_P.$ 
\end{lemma}
\begin{proof}
\begin{align*}
Z &= \star \frac{\gamma^{\wedge M}}{M!} = \star \frac{1}{M!}\sum_{(p_1, \ldots, p_M) \in K^M} \bigg\{ \prod_{i=1}^M \widehat m_{p_i} \bigg\} \epsilon_{p_1} \wedge \epsilon_{p_2} \wedge \cdots \wedge \epsilon_{p_M} \\
&= \frac{1}{M!}\sum_{(p_1, \ldots, p_M)\in K^M} \star(\epsilon_{p_1} \wedge \epsilon_{p_2} \wedge \cdots \wedge \epsilon_{p_M}) \bigg\{\prod_{i=1}^M \widehat m_{p_i} \bigg\}. 
\end{align*}
Conservation of momentum implies that $\star(\epsilon_{p_1} \wedge \epsilon_{p_2} \wedge \cdots \wedge \epsilon_{p_M}) = 0$ unless $(p_1, \ldots, p_M) \in K^M_0,$ in which case $\star(\epsilon_{p_1} \wedge \epsilon_{p_2} \wedge \cdots \wedge \epsilon_{p_M}) = C_{p_1, \ldots, p_M}$ and the lemma follows. 
\end{proof}

The structure coefficients are universal; external potential specific
information is only encoded by the product of shifted moments. While the
structure coefficients still require exterior algebra calculations, their
universality means they can be computed once and reused for all ensembles with
the same $L$ and $M$. 

The structure coefficients are invariant under permutation of the constituent
momentum modes and if desired could be uniquified by putting these momentum
modes into weakly increasing order. While these universal coefficients greatly
simplify computations, they are not algebraically independent. They are bound by
deep geometric constraints that will ultimately dictate how momentum is
transported through the ensemble.

\subsection{Momentum Pl\"ucker Relations}
\label{subsec:momentum-algebra.momentum-plucker}

The fundamental constraint on the momentum algebra arises from the simple fact
that the algebraic avatar of a particle, $\omega(z)$, is an $L$-blade. By the
antisymmetry of the exterior algebra, its wedge product with itself must
identically vanish:
\begin{equation}
\label{eq:omega-squared}
\omega(z) \wedge \omega(z) = 0.
\end{equation}
Expanding this single exterior algebra identity in the momentum basis generates
a strict family of quadratic dependencies among the momentum modes. We refer to
these generating constraints as the {\em momentum Pl\"ucker relations}.

In momentum coordinates,
\[
\omega(z) \wedge \omega(z) = \sum_{p} \sum_{q} z^{2c_M + p+q} \epsilon_p \wedge \epsilon_q = z^{2 c_M} \sum_{n} z^n \sum_{p + q = n} \epsilon_p \wedge \epsilon_q = 0.
\]
We define the quadratic elements
\[
r_n = \sum_{p+q=n} \epsilon_p \wedge \epsilon_q.
\]
Then the {\em momentum Pl\"ucker relations} on the momentum algebra: $\{r_n = 0 \}.$ Higher momentum Pl\"ucker relations are produced by likewise expanding $\omega(z)^{\wedge j} = 0$ with $j \geq 2,$ 
\[
\sum_{p_1 + \cdots + p_j = n} \epsilon_{p_1} \wedge \epsilon_{p_2} \wedge \cdots \wedge \epsilon_{p_j} = 0.
\]

\begin{lemma}\label{lemma:Plucker}
Let $T = [T_{i-j}]_{i,j}$ be a Toeplitz matrix acting on $\mathcal S_{M+1}$ then,
\[
r_{n,T} := \sum_{p+q = n} \epsilon_p \wedge T \epsilon_q = 0.
\]
\end{lemma}
This says that the momentum Pl\"ucker relations are stable under one-sided (or two-sided for that matter) Toeplitz substitutions. This invariance under Toeplitz transformations will later allow the momentum Plücker relations to be transported to the time-deformed backgrounds, leading to the transport identities and Hirota equations.
\begin{proof}
\begin{align*}
r_{n,T} &= \sum_{p+q = n} \epsilon_p \wedge \sum_j T_{q-j} \epsilon_j \\
&= \sum_k T_k \sum_{p+q=n} \epsilon_p \wedge \epsilon_{k+q} \\
&= \sum_k T_k \sum_{p+q=n -k} \epsilon_p \wedge \epsilon_{q} \\
&= \sum_k T_k r_{n-k} = 0. \qedhere
\end{align*}
\end{proof}

Expanding this single exterior algebra identity in the momentum basis generates
a strict family of quadratic dependencies among the momentum modes, the momentum
Pl\"ucker relations. While these relations provide rigid static constraints on
the structure coefficients for a fixed number of particles, generating true
integrable hierarchies requires elevating them into dynamic identities that
transport momentum between systems of differing particle numbers.

\section{Momentum Insertion and Extraction Operators}
\label{sec:insertion-extraction}

Comparing backgrounds across different particle sectors requires algebraic
operators capable of inserting or extracting a particle from a given state. In
the exterior algebra, particle insertion is naturally implemented by wedge
multiplication. However, within the commutative momentum algebra, there is no
natural interior product or contraction operation available to perform particle
extractions. 

To circumvent this algebraic missing piece, we construct extraction operators
indirectly via adjunction. The strategy is to embed the momentum algebra for an
$M$-particle system into the larger algebra for $M+1$ particles. We then define
an extraction operator by demanding that removing a particle from a saturated
$(M+1)$-particle background reproduces the exact same determinantal evaluation
as inserting a particle into an $(M-1)$-particle background. This adjunction
procedure generates a {\em conjugate momentum spine}, providing the exact
algebraic analogue to annihilation operators in standard fermionic systems.

Because we will work with systems with variable particle number, we will write,
for instance $I_M, V_M$ and $\mathcal S_M$ to denote the indices, vector space
and momentum spine for the $M$ particle ensemble.  

We work in the momentum algebra associated to $M+1$ particles, and we suppose
that $\gamma_-$ and $\gamma_+$ are elements of $\mathcal S_{M+1}$. In this
setting, in order to ameliorate the proliferation of subscripts, we set
$\mathcal B$ to be the momentum algebra generated by $\mathcal S_{M+1}$ in
$\Lambda V_{M+1}.$ Notice that $I_M \subseteq I_{M+1}$ and we define the linear
functional $\star_M : \mathcal B \rightarrow \R$ supported on $\mathcal B_0^M$
satisfying 
\[
\star_M \mathbf e_{I_M} = 1. 
\]
This agrees with the natural functional on the determinantal line $\Lambda^M
V_M$ as embedded in $\Lambda^{M+1} V_{M+1}.$

We set 
\[
\Gamma_- = \frac{\gamma_-^{\wedge(M-1)}}{(M-1)!} \in \mathcal B^{M-1} \qq{and} \Gamma_+ = \frac{\gamma_+^{\wedge(M+1)}}{(M+1)!} \in \mathcal B^{M+1}.
\]
We can produce a background consistent with the insertion of a particle into
$\Gamma_-$ with momentum $p$ by wedging $\Gamma_-$ by $\epsilon_p.$ That is
$\epsilon_p \wedge \Gamma_- \in \mathcal B^M.$ We write $\mathcal S^{\wedge} =
\mathcal S_{M+1}$ and $\mathcal B^{\wedge} = \mathcal B$ when viewing elements
of the momentum algebra as particle insertion operators under wedge, $\alpha
\mapsto (\Xi \mapsto \alpha \wedge \Xi).$

We now construct extraction operators by requiring that inserting a particle
into an $(M-1)$-particle background produce the same evaluated
background as extracting a particle from an $(M+1)$-particle background.
That is, we look for an adjoint action $\xi_{-q}$ on $\Gamma_+$ which reproduces the background 
\[
\epsilon_q \wedge \frac{\gamma_+^{\wedge(M-1)}}{(M-1)!}.
\]
This adjunction property determines the extraction operators in terms of the
structure coefficients.

\begin{lemma}\label{adjunction}
There exists map $\xi_{-q} : \mathcal B^{M+1}_0 \rightarrow \mathcal B_0^M$ such that
\[
\star_M ( \xi_{-q} \Gamma_+ ) = \star_M \left(\epsilon_q \wedge \frac{\gamma_+^{\wedge (M-1)}}{(M-1)!} \right).
\]
\end{lemma}
\begin{proof}
We construct one such map by prescribing its action on the coefficient expansion of $\Gamma_+.$ 
\[
\star_M \left(\epsilon_q \wedge \frac{\gamma_+^{\wedge(M-1)}}{(M-1)!} \right) = \frac{1}{(M-1)!} \sum_{j_1 + \cdots + j_{M-1} = -q} C^{(M)}_{q, j_1, \ldots, j_{M-1}} \prod_{i=1}^{M-1} \widehat m_{+,j_i},
\]
where $\widehat m_{+, j_i}$ is the $j_i$-shifted moment of the measure determining $\gamma_+,$ and we have superscripted the structure coefficients to indicate they correspond to a system with $M$ particles. Looking at 
\[
\star_{M+1} \Gamma_+ = \frac{1}{(M+1)!} \sum_{p_0 + p_1 + \cdots + p_M=0} C^{(M+1)}_{p_0, \ldots, p_M} \prod_{i=0}^M \widehat m_{+, p_i},
\]
and hence, because $\star_{M+1}$ gives the coefficient of $\mathbf e_{I_{M+1}},$
\[
\Gamma_+ = \frac{\mathbf e_{I_{M+1}}}{(M+1)!} \sum_{p_0 + p_1 + \cdots + p_M=0} C^{(M+1)}_{p_0, \ldots, p_M} \prod_{i=0}^M \widehat m_{+, p_i}.
\]
Define the $\xi_{-q}$ by $\xi_{-q}(\mathbf e_{I_{M+1}}) = \mathbf e_{I_M}$ and
\[
\xi_{-q} \left(C^{(M+1)}_{p_0, \ldots, p_M}\right) = 
\frac{1}{M(M+1)} \frac{C^{(M)}_{q, p_1, \ldots, p_{M-1}}}{\widehat m_{+, p_0} \widehat m_{+, p_M}}
\]
which we extend by linearity on the coefficients of $\Gamma_+.$ Note that $\xi_{-q}$ acts on structure coefficients by choosing two particular momentum modes, $p_0$ and $p_M$ in the $M+1$ particle system, and replacing them with a single momentum mode $q.$ This provides a background-dependent map into the structure coefficients in the $M$ particle system. We remark that if $p_0 + p_M \neq q,$ then $C^{(M)}_{q, p_1, \ldots, p_{M-1}} = 0$ by momentum conservation.

That is,
\[
\xi_{-q} \circ \Gamma_+ := \xi_{-q}(\Gamma_+) = \frac{\mathbf e_{I_M}}{(M-1)!} \sum_{p_1 + \cdots + p_{M-1} = -q} C^{(M)}_{q, p_1, \ldots, p_{M-1}} \prod_{i=1}^{M-1} \widehat m_{+,p_i}.
\]
Application of $\star_M$ to both sides of this equation establishes the lemma.
\end{proof}

\begin{rem}
The choice of slots $0$ and $M$ in the definition of $\xi_{-q}$ is a gauge
choice: any pair of slots $(i,j)$ satisfying $p_i + p_j = q$ yields a map with
the same scalar output after applying $\star_M$, by the momentum
Pl\"ucker relations. To make the construction canonical one may symmetrize over
all $\binom{M+1}{2}$ such pairs, at the cost of a less explicit treatment of
the residual moment factors $\prod_{k \neq i,j} \widehat{m}_{+,p_k}$.
This non-uniqueness reflects the fact that $\xi_{-q}$ provides a
partial, background-dependent ``inversion'' of the Gram form $\gamma_+$, sufficient
for the purposes of deriving the Hirota identity but stopping short of a
canonical inverse; see Section~\ref{subsec:discussion.correlation-kernels}.
\end{rem}

We define the {\em conjugate momentum spine} $\mathcal S^{\vee}$ to be the span
of the $\xi_{-q}$ over all momentum modes. We call elements of $\mathcal
S^{\vee}$ {\em extractions}.

Insertion corresponds to wedging a momentum mode into a background, while
extraction is defined indirectly by requiring that it reproduce the effect of
removing a particle when evaluated on saturated backgrounds. In other words, extraction
operators are defined by adjunction with respect to the evaluation functional
rather than by contraction in the exterior algebra. Together, they form a pair 
of conjugate momentum spines capable of altering particle number in a controlled way.

\subsection{Momentum Transport Operators}
\label{subsec:insertion-extraction.transport-operators}

By acting simultaneously on a pair of independent backgrounds---inserting a particle 
into one while extracting a particle from the other---we can transfer momentum 
between systems while preserving their combined particle number. We refer to the 
tensor products that execute this transfer as {\em momentum transport operators}. 
These operators provide the exact algebraic mechanism needed to lift the static 
momentum Pl\"ucker relations into dynamic transport identities.

Toeplitz matrices act naturally on the momentum spine by shifting and weighting
momentum modes. Since the momentum Plücker relations are invariant under
Toeplitz substitutions, we may replace the insertion and extraction modes by
Toeplitz-weighted superpositions without destroying the algebraic relations.
Given a fixed Toeplitz matrix $T$ we define $\phi_{q} = T \epsilon_q \in
\mathcal S^{\wedge}$ and $\varphi_{-q} = T \xi_{-q} \in
\mathcal S^{\vee}.$ These are weighted superposition of particle insertions and
extractions. 

An element of $\mathcal S^{\wedge} \otimes \mathcal S^{\vee}$ acts on a pair of
backgrounds by inserting a particle into one background and extracting a
particle from the other. Tensor products of insertion and extraction operators
ultimately describe the transfer of momentum between particle systems while
preserving the total number of particles. In particular, the {\em one particle
$p,-q$-momentum transport} channel $\epsilon_p \otimes \varphi_{-q} \in \mathcal
S^{\wedge} \otimes \mathcal S^{\vee}$ is the operator acting on $\Gamma_-
\otimes \Gamma_+$ by
\[
(\epsilon_p \otimes \varphi_{-q})(\Gamma_- \otimes \Gamma_+) = (\epsilon_p \wedge \Gamma_-) \otimes (\varphi_{-q} \circ \Gamma_+) \quad \in \mathcal B^M \otimes \mathcal B^M.
\]
This deletes a particle with conjugate momentum $q$ from $\Gamma_+$ and inserts
a particle with momentum $p$ into $\Gamma_-$. 

To organize the transport operators by momentum transfer, we introduce
generating functions for insertion and extraction operators. By
\eqref{eq:omega-generating}, the generating function for insertions is simply
$\omega(z),$ while that for extractions is given by
\[
\Omega(z) = \sum_{q} z^{-c_M-q} \varphi_{-q}
\]
The (operator valued) generating function for one-particle momentum transport is given by
\[
(\omega(z) \otimes \Omega(z))(\Gamma_- \otimes \Gamma_+) = \sum_p \sum_q z^{p-q} (\epsilon_p \wedge \Gamma_-) \otimes (\varphi_{-q} \circ \Gamma_+).
\]
By constraining these transport channels with the momentum Pl\"ucker relations, we generate a rigid system of bilinear identities that govern the ensemble.

\subsection{Momentum Transport Relations}
\label{subsec:insertion-extraction.transport-relations}

While the conceptual mechanism of momentum transport is straightforward, the
full extraction algebra required to rigorously define it is highly complex, as
it must naturally accommodate multi-particle momentum exchanges. Fortunately,
the fundamental observables considered here only probe single-particle
extractions. We can therefore bypass the full algebraic complexity by
restricting our focus to a linear slice. By replacing the full extraction
algebra with an auxiliary structure that agrees with it exactly on this
slice---and in which the Pl\"ucker relations inherently vanish---we obtain a
highly efficient mathematical framework sufficient for deriving all one-particle
transport identities.

Define
\[
\kappa_n = \sum_{p-q=n} (\epsilon_p \otimes \varphi_{-q}) 
\]
Then, $\kappa_n$ is an operator acting on $\Gamma_- \otimes \Gamma_+,$ producing a background in $\mathcal B^M \otimes \mathcal B^M;$ 
\[
\kappa_n (\Gamma_- \otimes \Gamma_+) = \sum_{p-q=n}  (\epsilon_p \wedge \Gamma_- \otimes \varphi_{-q} \circ \Gamma_+).
\]
That is, $\kappa_n$ is the operator which transports net momentum $n$ from the right
background to the left via transport of one particle.

We look at observables that are supported on $\mathcal B_0^M \otimes \mathcal B_0^M$ via the bilinear functional $\Pi_M : \mathcal B \otimes \mathcal B \rightarrow [0, \infty)$ specified by
\[
\Pi_M (\Xi_1 \otimes \Xi_2) = \star_M \Xi_1 \cdot \star_M \Xi_2.
\]
The point is that the bilinear functional $\Pi_M$ detects only $M$-particle,
zero-momentum backgrounds. As a result, higher extraction terms and non-zero
momentum backgrounds are `off-shell' for $\Pi_M$, and only the momentum
conserving, linear transport channels are visible.

\begin{thm}[One-Particle Momentum Transport Relations]\label{thm:transport}
\[
\Pi_M \sum_{p-q=n} (\epsilon_p \otimes \varphi_{-q})(\Gamma_- \otimes \Gamma_+) = 0.
\]
\end{thm}
\begin{proof}
Since $\Pi_M$ detects only the degree $0$ and degree $1$ insertion-extraction
channels, higher degree extractions are invisible to $\Pi_M$. We have not
explicitly defined what a higher degree extraction is, but we need not. For our
purposes it is enough to observe that they act on $\Gamma_+$ by producing a
background in $\mathcal B_0^m$ for some $m < M,$ and hence are undetectable by $\Pi_M.$ 

We may therefore replace the full algebra by a square-zero algebra in which all
quadratic products vanish. Let $\mathcal F^{\wedge}$ be the square-zero algebra
generated by $\mathcal S^{\wedge}$ and let $\mathcal F^{\vee}$ be the
square-zero algebra generated by $\mathcal S^{\vee}.$ In these reduced algebras,
the momentum Plücker relations hold identically. That is, 
\[
\sum_{p+q=n} \epsilon_p \epsilon_q  = 0 \in \mathcal F^{\wedge} \qq{and} \sum_{-p-q=n} \varphi_{-p} \varphi_{-q} = 0 \in \mathcal F^{\vee}.
\]

As vector spaces $\mathcal S^{\wedge}$ and $\mathcal S^{\vee}$ are isomorphic,
and hence $\mathcal F^{\wedge} \cong \mathcal F^{\vee}.$ We may therefore define the
coproduct-like homomorphism $\Delta : \mathcal B \rightarrow \mathcal F^{\wedge}
\otimes \mathcal F^{\vee}$ by
\[
\Delta(\epsilon_p) = (\epsilon_p \otimes 1) + (1 \otimes \varphi_{-p}),
\] 
which we extend by linearity and multiplication in $\mathcal F^{\wedge} \otimes \mathcal F^{\vee}.$ Then, because $r_n = 0,$
\begin{align*}
0 = \Delta(r_n) &= \sum_{p+q = n} \big( (\epsilon_p \otimes 1) + (1 \otimes \varphi_{-p}) \big)\big( (\epsilon_q \otimes 1) + (1 \otimes \varphi_{-q}) \big) \\
&= \sum_{p+q=n} (\epsilon_p \epsilon_q \otimes 1) + 2 \sum_{p-q=n} (\epsilon_p \otimes \varphi_{-q}) + \sum_{-p-q=n} (1 \otimes \varphi_{-p} \varphi_{-q}) \\
&= 2 \kappa_n,
\end{align*}
In particular,
\[
\kappa_n = \sum_{p-q=n} (\epsilon_p \otimes \varphi_{-q}) = 0 \quad \in \mathcal F^{\wedge} \otimes \mathcal F^{\vee}.
\]
This implies 
\[
\Pi_M \kappa_n (\Gamma_- \otimes \Gamma_+)  = \frac12 \Pi_M \Delta(r_n)(\Gamma_- \otimes \Gamma_+) = 0.
\]
Thus the one-particle momentum transport relations follow from the momentum
Plücker relations after projecting to the linear insertion-extraction slice
visible to $\Pi_M$. 
\end{proof}

To bridge the gap between these discrete algebraic transport identities and the
continuous differential equations of classical integrable systems, we must now
package these transport operations into generating functions.

\section{Hyperpfaffian Tau Functions and Deformed Miwa Coordinates}
\label{sec:tau-functions}

This generating function formulation is achieved by introducing continuous time
deformations into the external potential of the log-gas. By allowing the
underlying weight function to flow along an infinite sequence of auxiliary
times, the discrete algebraic acts of particle insertion and extraction are
smoothly translated into continuous shifts of the background moments. Under this
deformation, the static hyperpfaffian partition functions evolve into dynamic
$\tau$-functions. By repackaging these infinite time shifts via Miwa
coordinates, the transport operators can be evaluated against a spectral
parameter, allowing the structural identities of the momentum algebra to
naturally take the form of Hirota bilinear equations.

Let $\mathbf t = (t_0, t_1, \ldots)$ be a sequence of {\em times}. The time
variables provide coordinates on the momentum spine, and shifts in the
times correspond to insertion and extraction operations in the momentum algebra.
That is, if we define the weight
\[
w(\mathbf t; x) = \exp\left\{\sum_{i=0}^\infty t_i x^i\right\},
\]
and measure 
\[
\mu(\mathbf t; dx) = w(\mathbf t; x) \, dx,
\]
then we may define the Gram form $\gamma(\mathbf t)$ at time $\mathbf t$ by
\[
\gamma(\mathbf t) = \sum_p \widehat m_p(\mathbf t) \epsilon_p \qq{where} \widehat m_p(\mathbf t) = \int_{\R} x^{c_M + p} \, d\mu(\mathbf t; x)
\]

The $\tau$-function is the partition function for the $M$-particle ensemble with weight $w(\mathbf t; x),$
\[
\tau_{M}(\mathbf t) = \star_{M} \frac{\gamma(\mathbf t)^{\wedge M}}{M!}.
\]

To express these dynamic $\tau$-functions within the standard framework of
integrable systems, we must systematically package their infinite time shifts.
This parameterization will ultimately allow the momentum transport relations to
be cast as bilinear residue identities.

\subsection{Deformed Miwa Coordinates}
\label{subsec:tau-functions.miwa}

We achieve this parameterization by introducing Miwa coordinates, which 
encode discrete shifts across the infinite sequence of time variables 
into a single generating function. In this framework, negative translations of 
the Miwa times correspond precisely to the momentum shifts of particle insertions. 
However, because the underlying momentum algebra is bigraded by both momentum 
and particle number, a single spectral parameter is insufficient to capture the 
full dynamics. We must also introduce an additional fugacity parameter to 
rigorously track the changes in particle number associated with extractions. 
By mapping our discrete algebraic operators onto these continuous, two-parameter 
deformations, we demonstrate that the insertion and extraction generating 
functions constructed earlier are not ad hoc constructs, but arise exactly 
as the natural time-shifted $\tau$-functions of the ensemble.

Given spectral parameter $z$ we define the {\em Miwa coordinate} \cite{Miwa1982},
\[
[z] := \left(\frac{z^i}{i} \right)_{i \geq 1}.
\]
Note that 
\[
w(\mathbf t - L^2[z^{-1}]; x) = \left(1 - \frac{x}{z}\right)^{L^2} \, w(\mathbf t; x) = z^{-L^2} (z - x)^{L^2} w(\mathbf t; x).
\]
We introduce the time deformations $\boldsymbol \delta(u,z) = (\delta_k(u,z), k \geq 1)$ by
\[
\sum_{k \geq 1} \delta_k(u, z) x^k = \log \left((1-u) + u \left(1 - \frac{x}{z}\right)^{-L^2} \right).
\]
These are modified Miwa coordinates that introduce both a spectral variable $z,$ which keeps track of momentum, and a fugacity variable $u$ which tracks number of extractions used in momentum transport.

A negative Miwa shift by $\boldsymbol \delta(1,z)$ inserts a particle at spectral parameter $z.$
\begin{lemma}\label{lemma:negative-miwa}
Let ${\displaystyle \Gamma_-(\mathbf t) = \frac{\gamma(\mathbf t)^{\wedge(M-1)}}{(M-1)!},}$ then, 
\[
w(\mathbf t; z) \frac{\tau_{M-1}(\mathbf t - \boldsymbol{\delta}(1, z))}{\tau_M(\mathbf t)} = z^{-L^2(M-1)} \star_M \!\! \left( \omega(z) \wedge \Gamma_-(\mathbf t) \right).
\]
\end{lemma}
\begin{proof}
For the negative translation by $\boldsymbol{\delta}(1,z),$ We note that
\[
\sum_{k \geq 1} \delta_k(1, z) x^k = -L^2 \log\left(1 - \frac{x}{z} \right),
\]
and hence
\[
\tau_{M-1}(\mathbf t - \boldsymbol{\delta}(1, z)) = \tau_{M-1}(\mathbf t - L^2[z^{-1}]).
\]
It follows that
\begin{align*}
\tau_{M-1}(\mathbf t - L^2[z^{-1}]) &= \frac{z^{-L^2(M-1)}}{(M-1)!} \int_{\R^{M-1}} \prod_{m=1}^{M-1} (z - x_m)^{L^2} \cdot \prod_{j < k} (x_k - x_j)^{L^2} \, d\mu^{M-1}(\mathbf t; \mathbf x) \\
&= \frac{\tau_M(\mathbf t)}{w(\mathbf t; z)} z^{-L^2(M-1)} R_1(\mathbf t; z),
\end{align*}
where $R_1(\mathbf t; z)$ is the first correlation function in the $M$-particle ensemble with weight $w(\mathbf t; x).$ By Lemma~\ref{lemma:correlations},
\[
w(\mathbf t; z) \frac{\tau_{M-1}(\mathbf t - L^2[z^{-1}])}{\tau_M(\mathbf t)} = z^{-L^2(M-1)} \star_M \left( \omega(z) \wedge \Gamma_-(\mathbf t)\right). \qedhere
\]
\end{proof}

For $|q| \leq c_{M+1},$ define
\[
\varphi_{-q}(\mathbf t') = {L^2 + c_M + q - 1 \choose c_M + q} \sum_p \widehat m_{c_M + q - p}(\mathbf t') \xi_{-p} \qquad \in \mathcal S^{\vee}.
\]
The $\varphi_{-q}$ form a basis for $\mathcal S^{\vee}$ generated by
\[
\Omega(\mathbf t'; z) = \sum_{q} z^{-q-c_M} \varphi_{-q}(\mathbf t').
\] 

A positive translation by the Miwa coordinate $\boldsymbol{\delta}(u,z)$
generates momentum extractions, and the coefficient of $u^1$ corresponds to
momentum extractions given by the transport of one particle. 
\begin{lemma}\label{lemma:positive-miwa}
Let ${\displaystyle \Gamma_+(\mathbf t') = \frac{\gamma(\mathbf t')^{\wedge(M+1)}}{{(M+1)!}},}$ then,
\[
[u^1] \tau_{M+1}(\mathbf t' + \boldsymbol{\delta}(u,z)) = \star_M \left(\Omega(\mathbf t'; z) \circ \Gamma_+ \right).
\]
\end{lemma}
\begin{proof}
For the positive translation we write
\begin{align*}
w(\mathbf t' + \boldsymbol{\delta}(u,z); x) &= \left( (1-u) + u \left(1 - \frac{x}{z}\right)^{-L^2} \right) w(\mathbf t'; x) \\
&= \left(1 + u \sum_{j \geq 1} {L^2 + j - 1 \choose j} \left(\frac{x}{z}\right)^j \right) w(\mathbf t'; x), 
\end{align*}
and
\[
\gamma(\mathbf t' + \boldsymbol{\delta}(u,z)) = \gamma(\mathbf t') + u \sum_{j \geq 1} {L^2 + j - 1 \choose j} z^{-j} \int_{\R} \omega(x) x^j \, d\mu(\mathbf t'; x).
\]
We define
\[
\phi_j(\mathbf t') = {L^2 + j - 1 \choose j} \sum_{p} \widehat m_{j+p}(\mathbf t') \epsilon_p; \qquad j \geq 0.
\]
so that
\[
\gamma(\mathbf t' + \boldsymbol{\delta}(u,z)) = \gamma(\mathbf t') + u \sum_{j > 0} z^{-j} \phi_j(\mathbf t').
\]
It follows that
\begin{align*}
\tau_{M+1}(\mathbf t' + \boldsymbol{\delta}(u,z)) &= \star_{M+1} \frac{1}{(M+1)!}\left(\gamma(\mathbf t') + u \sum_{j > 0} z^{-j} \phi_j(\mathbf t')\right)^{\wedge(M+1)} \\
&= \tau_{M+1}(\mathbf t') + \star_{M+1} \left(\sum_{\ell=1} \frac{u^{\ell}}{\ell!} \sum_{k > 0} z^{-k} \sum_{p_1 + \cdots + p_{\ell} = k} \bigg\{ \bigwedge_{i=1}^{\ell} \phi_{p_i}(\mathbf t') \bigg\} \wedge \frac{\gamma(\mathbf t')^{\wedge(M+1-\ell)}}{(M+1 - \ell)!} \right).
\end{align*}
The coefficient of $u^1$ is thus given by
\[
[u^1] \tau_{M+1}(\mathbf t' + \boldsymbol{\delta}(u,z)) = \sum_{k > 0} z^{-k} \star_{M+1} \left(\phi_k(\mathbf t') \wedge \frac{\gamma(\mathbf t')^{\wedge M}}{M!}\right).
\]
Next we use Lemma~\ref{adjunction} to write
\begin{align*}
\star_{M+1} \left(\phi_k(\mathbf t') \wedge \frac{\gamma(\mathbf t')^{\wedge M}}{M!}\right) &= {L^2 + k - 1 \choose k} \sum_p \widehat m_{k+p} \star_{M+1} \left(\epsilon_p \wedge \frac{\gamma(\mathbf t')^{\wedge M}}{M!} \right) \\
&= {L^2 + k - 1 \choose k} \sum_p \widehat m_{k+p}(\mathbf t') \star_{M} (\xi_{-p} \circ \Gamma_+(\mathbf t'))
\end{align*}
We conclude, 
\[
[u^1] \tau_{M+1}(\mathbf t' + \boldsymbol{\delta}(u,z)) = \star_M \left(\Omega(\mathbf t'; z) \circ \Gamma_+ \right). \qedhere
\]
\end{proof}

We have therefore established a complete dictionary between the discrete
algebraic algebra and the continuous integrable flows: negative Miwa
translations generate momentum insertions, while positive translations,
precisely tracked by the fugacity variable, generate extractions. Consequently,
the generating functions for these transport operators are not merely formal
power series; they are exactly the time-deformed $\tau$-functions of the
ensemble.

\section{The Hirota Bilinear Identity}
\label{sec:hirota}

By packaging the momentum transport relations into these continuous Miwa
coordinates, the generating functions for particle insertion and extraction
naturally take the form of primal and dual Baker--Akhiezer wave functions
\cite{AlexandrovZabrodin2013, MiwaJimboDate2000}. When the discrete one-particle
transport identities are evaluated using these wave functions, the geometric
constraints of the momentum Pl\"ucker relations condense into a single bilinear
residue identity. This identity is precisely the Hirota equation, which
completely characterizes the integrable hierarchy underlying the charge-$L$
ensemble.

We define the (primal) Baker--Akhiezer wave function 
\begin{align*}
\psi^-(\mathbf t; z) &:= z^{2 c_M} w(\mathbf t; z) \frac{\tau_{M-1}(\mathbf t - \boldsymbol \delta(1,z))}{\tau_M(\mathbf t)} \\
&= \star_M \left( \omega(z) \wedge \Gamma_- \right) \\
&= \sum_{p} z^{p + c_M} \star_M\left(\epsilon_p \wedge \Gamma_-\right).
\end{align*}
Similarly, the conjugate Baker-Akhiezer wave function is 
\begin{align*}
\psi^+(\mathbf t'; z) &:= [u^1] \tau_{M+1}(\mathbf t' + \boldsymbol{\delta}(u,z)) \\ 
&=  \sum_{q} z^{-q - c_M} \star_M \left(\varphi_{-q}(\mathbf t') \circ \Gamma_+ \right).
\end{align*}
We normalize the Baker–Akhiezer wave functions so that the generating series for insertions and extractions are Laurent series in $z$ whose constant term corresponds to momentum conservation.

Then,
\begin{align*}
\psi^-(\mathbf t; z) \psi^+(\mathbf t'; z) &= \sum_p \sum_q z^{p-q} \star_M\left(\epsilon_p \wedge \Gamma_- \right) \star_M \left(\varphi_{-q}(\mathbf t') \circ \Gamma_+ \right).
\end{align*}

The function $\psi^-(\mathbf t; z)$ is the generating function for particle
insertions into the $(M-1)$-particle background, while $\psi^+(\mathbf t'; z)$
is the generating function for momentum extractions from the $(M+1)$-particle
background. The product $\psi^-(\mathbf t; z)\psi^+(\mathbf t'; z)$ encodes all
one-particle momentum transport channels. Taking the coefficient of $z^0$
selects transport channels with zero net momentum transfer, and the momentum
transport relations imply that the sum of these channels vanishes. This is the
Hirota bilinear identity for hyperpfaffian $\tau$-functions in charge-$L$
ensembles. 
\begin{thm}
\[
[z^0] \left(\psi^-(\mathbf t; z) \psi^+(\mathbf t'; z) \right)= 0.
\]
We may renormalize the integrand by any scalar without disturbing
the relation. In particular, if more symmetry is desired between the positive
and negative translations of the $\tau$-functions, we may write,
\[
\oint z^{2 c_M} w(\mathbf t; z) \frac{\tau_{M-1}(\mathbf t - \boldsymbol \delta(1,z))}{\tau_M(\mathbf t)} \cdot [u^1] \frac{\tau_{M+1}(\mathbf t' + \boldsymbol{\delta}(u,z))}{\tau_M(\mathbf t')} \frac{dz}{z} = 0.
\]
\end{thm}
\begin{rem} 
The contour integral here is purely formal coefficient extraction, and meant
only to be suggestive of the classical form of the Hirota equation. 
\end{rem}
\begin{proof}
\begin{align*}
[z^0] \left(\psi^-(\mathbf t; z) \psi^+(\mathbf t'; z) \right) &= \sum_{p-q=0} \star_M\left(\epsilon_p \wedge \Gamma_- \right) \star_M \left(\varphi_{-q}(\mathbf t') \circ \Gamma_+ \right) \\
&= \Pi_M \sum_{p-q=0} (\epsilon_p \otimes \varphi_{-q}(\mathbf t'))(\Gamma_-(\mathbf t) \otimes \Gamma_+(\mathbf t')) = 0,
\end{align*}
by Theorem~\ref{thm:transport}.
\end{proof}

The Hirota bilinear equation ultimately arises from the single exterior algebra 
identity $\omega(z)\wedge\omega(z)=0$. This integrable structure is uncovered 
by expressing this geometric constraint in momentum coordinates, lifting it 
into dynamic transport between particle backgrounds via conjugate insertion 
and extraction operators, and finally packaging these operations as Miwa 
time deformations of the $\tau$-function.

\section{Discussion and Outlook}
\label{sec:discussion}

\subsection{Relation to Sato Theory and Integrable Hierarchies}
\label{subsec:discussion.sato-theory}

In the classical Sato theory, integrable hierarchies such as the
Kadomtsev--Petviashvili (KP) and Toda lattices are understood as dynamical
systems on an infinite-dimensional Grassmann manifold \cite{MiwaJimboDate2000,
Sato1981}. In that setting, the Hirota bilinear equations emerge directly from
the Pl\"ucker relations that define the Grassmannian embedding. In the present
work, the integrable structure of the charge-$L$ ensemble arises from an
analogous algebraic mechanism, but in a strictly finite-dimensional setting. The
dimensional reduction provided by the momentum algebra generates a finite
sequence of momentum Pl\"ucker relations originating from the blade identity
$\omega(z) \wedge \omega(z) = 0$.

A natural question is how this finite-dimensional momentum algebra maps onto the
standard classification of integrable hierarchies. Because the saturated
backgrounds possess a strict momentum conservation law (total momentum must be
zero), the associated $\tau$-functions satisfy rigid truncation conditions with
respect to the Miwa time shifts. Fully embedding these truncated
$\tau$-functions into the standard infinite-dimensional hierarchies, or
classifying the specific hierarchy reductions they represent for general $L$,
remains an open problem.

\subsection{Asymptotics and Large Particle Number Limits}
\label{subsec:discussion.asymptotics}

The results presented in this work focus exclusively on the exact algebraic
structure of the charge-$L$ ensemble for finite particle number $M$. The
momentum algebra provides a significant computational simplification in this
finite regime, reducing the degrees of freedom required to describe saturated
backgrounds from the combinatorial dimension $\binom{LM}{L}$ down to the linear
dimension $L^2(M-1)+1$. 

However, taking the thermodynamic limit $M \to \infty$ fundamentally alters the
mathematical landscape. In classical random matrix theory, such large-$M$ limits
lead to universal scaling regimes governed by the sine, Airy, or Bessel kernels.
In the present framework, as $M$ grows, the dimension of the momentum spine
diverges. Consequently, the finite-dimensional Sato analogue described here must
eventually transition into a true infinite-dimensional Grassmannian formulation. 

Extracting the asymptotic behavior of the hyperpfaffian $\tau$-functions in this
limit requires moving beyond the purely algebraic methods developed here and
entering the realm of asymptotic analysis. This would likely require developing
appropriate Riemann--Hilbert problems, loop equations, or large-$M$ saddle-point
approximations based on the shifted moments $\widehat m_p$. Furthermore, the
finite-dimensional momentum transport operators $\kappa_n$ defined in
Section~\ref{subsec:insertion-extraction.transport-relations} would
correspondingly need to be reformulated as continuous transport equations. We
leave the analytic challenge of passing the momentum Pl\"ucker relations and
Hirota equations through the large-$M$ limit as an open problem.

\subsection{Odd $L$}
\label{subsec:discussion.odd-L}

This work restricted attention to even $L$ to ensure that the $L$-forms
$\omega(x)$ commute, allowing the absolute values in the Gibbs measure to be
dropped without introducing alternating signs. However, the underlying
integrable structure is not strictly limited to the even $L$ regime, and
analogous Hirota formulas are expected to exist for odd $L$.

The primary obstacle for odd $L$ is combinatorial bookkeeping rather than a
fundamental algebraic barrier. When $L$ is odd, the $L$-forms anticommute,
meaning the algebraic avatars of the particles satisfy fermionic rather than
bosonic statistics. One natural algebraic strategy to resolve this is to
``double up'' the particles by wedging pairs of $L$-forms together to create
$2L$-forms. Because $2L$ is always even, these composite $2L$-forms commute,
effectively mapping the system back into a commutative subalgebra where a
modified adjunction and transport protocol can be applied. 

In this doubled regime, the integrals defining the Gram form will necessarily
pair particles together. Consequently, the analytic information of the odd-$L$
ensembles will be encoded in the skew moments of the measure, mirroring the
Pfaffian structure of the classical orthogonal ($\beta=1$) ensembles. Deriving
the explicit momentum Pl\"ucker relations and Hirota identities for this
skew-symmetric setting requires a modified momentum algebra, which we defer to
future work.

\subsection{Relationship with BKP}
\label{subsec:discussion.bkp}

It is a well-established feature of integrable probability that while $\beta=2$
ensembles are governed by determinantal point processes and the
Kadomtsev--Petviashvili (KP) hierarchy, the $\beta=1$ and $\beta=4$ ensembles
are governed by Pfaffian structures and are closely tied to the BKP hierarchy
\cite{MiwaJimboDate2000}. Our exterior algebra formulation provides a natural
geometric explanation for this algebraic shift. 

In the present framework, the $\beta=4$ ensemble corresponds exactly to the
$L=2$ case. When $L=2$, the fundamental algebraic avatars $\omega(x)$ are
2-forms. The hyperpfaffian of a Gram 2-form, as defined in
Section~\ref{subsec:exterior-algebra.partition-function}, reduces precisely to
the classical Pfaffian of a skew-symmetric matrix. Consequently, the momentum
Pl\"ucker relations governing the $L=2$ momentum algebra are exactly the
fundamental Pfaffian identities. 

Because the BKP hierarchy is ultimately generated by these Pfaffian algebraic
relations, the Hirota bilinear equations derived here for $L=2$ must project
directly onto the BKP hierarchy. Furthermore, for general even $L > 2$, the
hyperpfaffian structure generalizes the classical Pfaffian, suggesting that the
corresponding integrable structure is a higher-order ``hyper-BKP'' hierarchy.
Establishing the explicit dictionary between the Miwa-shifted momentum transport
operators constructed here and the neutral fermion vertex operators
traditionally used to formulate the BKP hierarchy remains an illuminating
direction for future research.

\subsection{Correlation Functions and Kernels}
\label{subsec:discussion.correlation-kernels}

While Lemma~\ref{lemma:correlations} provides an exact algebraic expression for
the correlation functions via insertion operators acting on the Gram background,
we have not yet extracted explicit analytic correlation kernels (such as the
determinantal or Pfaffian kernels seen in the classical $\beta = 1, 2, 4$
ensembles). In the classical setting, extracting the correlation kernel
typically relies on expressing the inverse of the moment matrix in terms of
orthogonal or skew-orthogonal polynomials, leading to a Christoffel--Darboux
formula. 

In the present exterior algebra framework, the analogous procedure would require
formally ``inverting'' the Gram form $\gamma$. The precise algebraic mechanism
for this inversion within the momentum algebra remains an open problem. However,
we hope that the full extraction algebra---extending beyond the linear
insertion-extraction slice $\Pi_M$ utilized in Theorem~\ref{thm:transport} to
derive the Hirota equations---will provide an alternative pathway to construct
these kernels. Utilizing higher-order extraction operators to bypass the direct
inversion of the Gram form was, in fact, the primary motivation for developing
the algebraic machinery presented here, and it remains a central goal for future
work.

\subsection{Circular Ensembles}
\label{subsec:discussion.circular}

While the exposition in this paper focused on log-gases on the real line, the
algebraic framework developed here applies equally well to circular ensembles
with $\beta = L^2$ even. In that setting, the shifted moments $\widehat m_p$ of the
measure on the real line are simply replaced by the Fourier coefficients of the
weight function on the unit circle. The underlying exterior algebra, the
confluentization process into $L$-blades, and the momentum Pl\"ucker relations
remain completely identical. The primary reason for restricting the present text
to the real line is one of exposition: presenting both situations simultaneously
introduces significant notational friction due to the divergent natural indexing
conventions required for the real line versus the circle.

Despite this presentational hurdle, circular ensembles are in many ways the
optimal setting for performing explicit algebraic calculations within this
framework. The rotational symmetry of the unit circle drastically simplifies the
structure of the momentum modes. In the fundamental case where the potential is
the uniform (Haar) measure on the unit circle, all non-zero Fourier coefficients
vanish. Consequently, the Gram form $\gamma$ collapses to a single momentum
mode. This massive reduction in complexity suggests that the circular ensembles
may be the most analytically tractable arena for computing explicit
hyperpfaffian correlation kernels and solving the associated transport
hierarchies in future work.

\subsection{Further Directions}
\label{subsec:discussion.further-directions}

The framework developed here opens several natural avenues for generalization. A
direct extension is the study of multi-species log-gas ensembles consisting of
particles with different integer charges, such as a mixture of charge-$L_1$ and
charge-$L_2$ particles. In the exterior algebra, this corresponds to backgrounds
formed by wedging $L_1$-blades with $L_2$-blades. While the associated momentum
algebra would necessarily be more complex---requiring a multi-graded structure
to track independent momentum conservation for each species---previous work on
mixed charge ensembles suggests these systems remain algebraically solvable
\cite{RiderSinclairXu2013}. We anticipate that charge and momentum-conserving
transport relations for these multi-species ensembles can be systematically
derived using the Berezin calculus.

On a broader structural level, the algebraic machinery constructed in this work
exhibits the distinct hallmarks of a discrete field theory. The generating
functions for momentum insertion and extraction, $\omega(z)$ and $\Omega(z)$,
behave analogously to vertex operators, and the transport identities they
satisfy are highly reminiscent of operator product expansions. Formalizing the
conjugate momentum spine and its transport relations into a rigorous vertex
operator algebra (VOA) would not only clarify the representation-theoretic
underpinnings of the charge-$L$ ensembles, but also provide a direct algebraic
bridge to the infinite-dimensional symmetries that govern classical integrable
hierarchies.

\bibliography{references}

\begin{center}
\noindent\rule{4cm}{.5pt}
\vspace{.25cm}

\noindent {\sc \small Christopher D.~Sinclair}\\
{\small Department of Mathematics, University of Oregon, Eugene OR 97403} \\
email: {\tt csinclai@uoregon.edu}
\end{center}

\end{document}